\documentclass[%
pra,
amsmath,amssymb,
aps,twocolumn,10pt
]{revtex4-2}

\usepackage{changepage}

\usepackage{amsthm}
\usepackage{array}

\usepackage[linesnumbered,ruled,vlined]{algorithm2e}

\usepackage{tikzpeople}
\usetikzlibrary{arrows.meta, positioning}
\usepackage[skins,breakable]{tcolorbox}
\usepackage{booktabs}
\usepackage{lipsum, babel}
\usepackage{adjustbox}

\usepackage{tikz}
\usepackage{amsmath}
\usetikzlibrary{matrix,fit,backgrounds}

\usepackage{amsmath} 

\usepackage{quantikz}
\usepackage{physics}
\usepackage{graphicx}
\usepackage{dcolumn}
\usepackage{bm}
\usepackage{hyperref}
\usepackage[capitalize]{cleveref}
\usepackage{hyperref}

\newtheorem{theorem}{Theorem}     
\newtheorem{lemma}{Lemma}         
\newtheorem{definition}{Definition} 
\newtheorem{protocol}{Protocol}

\newtcolorbox[auto counter]{example}[1][]{enhanced,fonttitle=\sffamily\bfseries\large,valign=center
,title={Example \thetcbcounter},label=#1,left=2pt,right=2pt,breakable}
\crefname{tcb@cnt@example}{Ex.}{Exs.}

\DeclareFontFamily{U}{bbold}{}
\DeclareFontShape{U}{bbold}{m}{n}
{
<-5.5> s*[1.069] bbold5
<5.5-6.5> s*[1.069] bbold6
<6.5-7.5> s*[1.069] bbold7
<7.5-8.5> s*[1.069] bbold8
<8.5-9.5> s*[1.069] bbold9
<9.5-11> s*[1.069] bbold10
<11-15> s*[1.069] bbold12
<15-> s*[1.069] bbold17
}{}

\DeclareRobustCommand{\identity}{%
	\text{\usefont{U}{bbold}{m}{n}1}%
	}

	\forceredefine{}
	
\begin{document}


	\title{Communication-Optimal Blind Quantum Protocols}

	\author{Ethan Davies}
	\email{Ethan.Davies.2021@live.rhul.ac.uk}
	\author{Alastair Kay}%
	\email{Alastair.Kay@rhul.ac.uk}
	\affiliation{Royal Holloway University of London, Egham, Surrey, TW20 0EX, UK}%

	\date{\today}

	\begin{abstract}
	A user, Alice, wants to get server Bob to implement a quantum computation for her. However, she wants to leave him blind to what she's doing. What are the minimal communication resources Alice must use in order to achieve information-theoretic security? In this paper, we consider a single step of the protocol, where Alice conveys to Bob whether or not he should implement a specific gate. We use an entropy-bounding technique to quantify the minimum number of qubits that Alice must send so that Bob cannot learn anything about the gate being implemented. We provide a protocol that saturates this bound. In this optimal protocol, the states that Alice sends may be entangled. For Clifford gates, we prove that it is sufficient for Alice to send separable states.
	\end{abstract}

	\maketitle
	
	Quantum computers promise speed-ups over classical computers for important computational tasks, ranging from quadratic to exponential \cite{shor1994,grover1996fast,harrow2009quantum}. This is anticipated to create a vast demand for the computational abilities provided by quantum computers, even among those who lack the expertise to manage them. The future will likely consist of parties (Alice) with small or no quantum capabilities wishing to gain the results of some quantum computation. To achieve this, they must interact with another agent (Bob) who has large quantum capabilities. For security reasons, Alice may want Bob to remain blind to the details of what she is computing; to hide both the computation and its outcome from Bob. Many protocols achieve this blind quantum computation \cite{childs,broadbentUniversalBlindQuantum2009,fitzsimons2017}.

	Blind quantum computing protocols can be separated into three main categories. In its primary mode, Alice has some limited quantum capability, known as semi-quantum. This is usually the ability to prepare a specific set of quantum states \cite{broadbentUniversalBlindQuantum2009} (a form of remote state preparation \cite{bennett2001remote}), or to apply a limited set of quantum gates \cite{childs}. Alice can be made classical, either by trading information theoretic security for computational security \cite{mahadev2020,brakerski2018,davies2024efficient}, or by playing multiple servers against each other in an interactive proof system \cite{reichardt2013classical}, with the unverifiable assumption that the otherwise untrusted servers do not communicate with each other.
	
	In this paper, we focus on two forms of semi-quantum capabilities for Alice. The first, Prepare and Send (PS), is the primary mode under which blind quantum computation is traditionally considered -- Alice prepares quantum states and sends them to Bob. For the second form, Receive and Measure (RM), it is Bob that prepares arbitrary quantum states and sends them to Alice. All Alice has to do is measure the qubits when she receives them \cite{morimae2012blind, morimae2014verification}. Destructive measurements will suffice. These two cases are largely equivalent. In particular, following either PS or RM, if the quantum communication is half a maximally entangled state, it can teleport a state in the opposite direction \footnote{The caveat being the consequences of different measurement outcomes in the teleportation. In our optimal protocols, it will turn out that these different outcomes precisely provide the padding we need to keep protocols secret, but we cannot say that about a generic protocol. Hence that \emph{largely} equivalent claim.}.

	We show how to optimise the protocols so that Alice makes the most of the resources she has (i.e.\ the amount of quantum communication between her and Bob). There are two extremes of what precisely we might optimise. In \cite{mantri2013}, they allowed a PS protocol where Alice could send a fixed number of qubits, $n$, to Bob. They then demanded the largest family of gates $\mathcal{F}$ that can be blindly achieved, providing (non-tight) upper and lower bounds. While we will tighten these bounds, our main focus is the opposite extreme: Alice has a set of gates $\mathcal{F}$ she would like to implement blindly. For a fixed $\mathcal{F}$, how few qubits can Alice prepare and send to still be able to blindly implement every gate in $\mathcal{F}$?

	We prove a lower bound for all protocols where Alice has either PS or RM under the assumption that the output state has a particular form of padding known as Pauli padding. We show that this bound is saturated when Alice can create entangled states or measure in an entangled basis. In the special case where $\mathcal{F} = \{\identity,U_{\text{Cl}}\}$ for a Clifford gate $U_{\text{Cl}}$, this bound can be met with Alice only being able to create or measure separable Pauli basis states, and $\mathcal{F}$ can be extended for free to a large set of different Clifford gates. We achieve this by identifying the Pauli operators that do not contribute to hiding the differences between the members of $\mathcal{F}$. These can be dropped from the output padding and, hence, Alice has to supply less entropy to Bob. 
\section{Summary of Results}

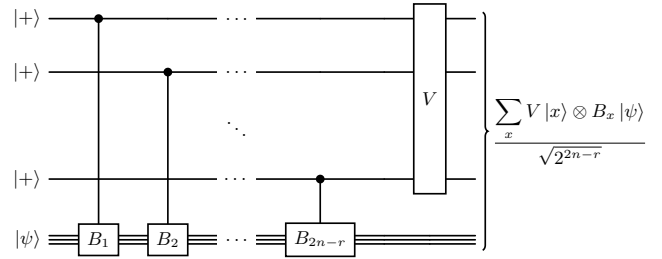
\begin{figure}
\centering
\begin{adjustbox}{max width=0.48\textwidth}
		\begin{quantikz}[wire types = {q,q,n,q,b}, classical gap = 0.07cm]
		\lstick{$\ket{+}$} & \ctrl{4} &&\gate[style={draw=none}]{\ldots} &&& \gate[4]{V}&\rstick[5]{$\displaystyle\frac{\displaystyle\sum_x V\ket{x}\otimes B_x \ket{\psi}}{\sqrt{2^{2n-r}}} $}\\
		\lstick{$\ket{+}$}&& \ctrl{3} & \gate[style={draw=none}]{\ldots}&& &&\\[-0.2cm]
		 & &  & \gate[style={draw=none}]{\ddots}  & & &&   \\[-0.2cm]
		\lstick{$\ket{+}$}&& & \gate[style={draw=none}]{\dots}  &\ctrl{1}& &&\\
		\lstick{$\ket{\psi}$} & \gate{B_1} & \gate{B_2} & \gate[style={draw=none}]{\ldots} & \gate{B_{2n-r}} & &&
		\end{quantikz}
\end{adjustbox}
\caption{General circuit used for RM optimal blind quantum computing. Bob entangles his state $\ket{\psi}$ with an additional $\text{dim}(\mathcal{B})=2n-r$ qubits which are then sent to Alice.}
\label{general_measurement_circuit}
\end{figure}

We briefly summarise the results here so that they do not get lost in the technical exposition that follows. Imagine that Alice wants to apply a gate $U\in\mathcal{F}$ on $n$ qubits, where Bob knows $\mathcal{F}$ but not the specific choice of $U$. There exists a subspace $\mathcal{P}_{\mathcal{F}}$ of the $n$-qubit Pauli operators $\mathcal{P}_n$ such that all members either commute or anti-commute with all members of $\mathcal{F}$:
\begin{align}
		\mathcal{P}_\mathcal{F} = \{P\in \mathcal{P}_n: UPU^\dagger = \pm P\quad\forall U \in \mathcal{F} \}.\label{eq:PF}
\end{align}
We define a second subspace $\mathcal{B}$ comprising the Paulis that commute with all members of $\mathcal{P}_\mathcal{F}$:
\begin{equation}
\mathcal{B} = \{B \in \mathcal{P}_n: BP=PB\quad \forall P \in \mathcal{P}_\mathcal{F}\}.\label{eq:B}
\end{equation}
This is a particularly useful space since all the unitaries in $\mathcal{F}$ have a decomposition in terms of it -- there exist coefficients $\gamma$ such that
$$
U=\sum_x\gamma_xB_x.
$$

Alice can achieve a blind implementation of $U\in\mathcal{F}$ by receiving $\text{dim}(\mathcal{B})=2n-\text{dim}(\mathcal{P}_\mathcal{F})$ qubits from Bob via the following protocol. This is optimal assuming a particular form of padding on the output state.

		\begin{protocol}\label{general_RM_protocol}
		\begin{enumerate}
			\item[]
			\item Alice and Bob agree upon a basis $\mathcal{B}$ for $\mathcal{F}$ and any unitary $V$ on $\text{dim}(\mathcal{B})$ qubits.
			\item Bob runs the circuit in \cref{general_measurement_circuit} and sends the top $\text{dim}(\mathcal{B})$ qubits to Alice.
			\item Alice chooses a $U \in \mathcal{F}$ that she would like to implement.
			\item Alice measures the qubits in the $\{V\ket{\phi_{B_xU}^*}\}$ basis where
			$$
\ket{\phi^*_U}=\sum_z\gamma^*_z\ket{z}
			$$
			for a unitary $U=\sum_z\gamma_zB_z$.
		\end{enumerate}
		\end{protocol}

After the protocol, Alice gets an answer $z\in\{0,1\}^{\text{dim}(\mathcal{B})}$, each equally likely, and hence knows that the gate $B_zU$ has been applied. She considers $B_z$ to be a padding that she adapts to in subsequent steps. Bob never learns $z$ and from his perspective, his initial state has been converted from $\ket{\psi}$ into $\sum_zB_z\proj{\psi}B_z$, independent of anything that Alice has done. We emphasise that Alice never sends anything, even classical information, to Bob once the protocol has begun, which severely limits Bob's attack channels. Simulations of the protocol have been implemented in the case where $\mathcal{F} = \{\identity, C\}$, for a Clifford $C$ \cite{Davies_Communication-Optimal_Blind_Quantum}.

\begin{example}[ex1]
		Let $U$ be the controlled-\textsc{not} gate. It acts on $n=2$ qubits and has a Pauli decomposition of
		$$
		U=\frac{1}{2}\left(\identity+Z_1+X_2-Z_1X_2\right).
		$$
		For the blind gate set $\mathcal{F}=\{\identity,U\}$, we can compute
		$$
		\mathcal{P}_\mathcal{F}=\{ \identity, Z_1, X_2, Z_1X_2  \}=\langle Z_1, X_2\rangle=\mathcal{B}.
		$$
		$\mathcal{P}_\mathcal{F}$ forms a vector space of dimension $r=\dim (\mathcal{P}_\mathcal{F})=2$. We claim (\cref{thm:resource,thm:opt_prot}) that optimal blind application of the controlled-\textsc{not} requires  $2n-r=2$ qubits of communication in either PS or RM. Since $U$ is a Clifford gate, the optimal protocol only uses separable states or single-qubit Pauli basis measurements (\cref{lem:sep_stab}).

		Specifically, Bob has a state $\ket{\psi}$ to which the controlled-\textsc{not} is to be applied. He runs the circuit
		
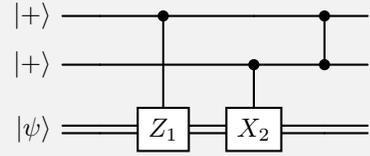
\[
		\begin{quantikz}[wire types={q,q,c}, classical gap = 0.05cm]
		\lstick{$\ket{+}$}&& \ctrl{2} && \ctrl{1}&\\
		\lstick{$\ket{+}$}&&& \ctrl{1}&\ctrl{0} &\\
		\lstick{$\ket{\psi}$}&& \gate{Z_1}&\gate{X_2}&&
		\end{quantikz}
		\]
		and sends the first two qubits to Alice. She either measures them in the $Z$ basis to realise an $\identity$ gate (up to measurement-result-dependent Pauli padding) or in the $X$ basis to realise the controlled-\textsc{not} gate.
		\end{example}

	\section{Pauli Padding}

	Our initial target is to blindly implement a particular gate $U$, i.e.\ to have possibly implemented $U^d$ for $d\in\{0,1\}$ and the goal is for Bob to be unable to determine $d$. This is equivalent to selecting a member of $\mathcal{F}=\{\identity,U\}$ to implement at random, leaving the server none the wiser as to which was chosen. Security will be achieved by causing Bob, who starts with state $\ket{\psi}$, to arrive at a state $VU^d\ket{\psi}$ where $V$ is some padding unitary that Alice knows (and can adapt to in any subsequent protocol). This padding is essential since if Bob were to cheat, he could supply any state $\ket{\psi}$ he wanted in place of the state Alice is expecting, and would be capable of distinguishing between $\identity\ket{\psi}$ and $U\ket{\psi}$ with some non-zero probability. By choosing the possible padding operators $\{V\}$, known as the padding set, such that
	$$
	\sum_Vp^{(\identity)}_VV\proj{\psi}V^\dagger=\sum_Vp^{(U)}_VVU\proj{\psi}U^\dagger V^\dagger
	$$
	for all $\ket{\psi}$, the two cases are indistinguishable to Bob.

	For a gate satisfying $U^k=\identity$ for some positive integer $k$, there is a trivial solution with no quantum communication -- Alice picks a random integer $j$ in the range 0 to $k-1$ and asks Bob to apply $U^j$ to the state $\ket{\psi}$ that he holds. She then interprets this as $U^{j-d}U^d\ket{\psi}$ where $U^{j-d}$ is a padding which must be adapted for in future computations.

	However, while we are calculating the communication bounds for a single gate implementation, we are ultimately interested in implementing a quantum circuit; a sequential application of multiple gates. It is at this point that our trivial scheme falls down. In order for Alice to keep track of the padding from each different gate, she must in fact implement a simulation of the entire quantum computation on her classical computer, rendering all speed-ups null and void.

	Instead, we will assume a specific form of padding, known as Pauli padding, that should be shared by every gate, and whose effects can be efficiently propagated through a circuit. There are certainly intermediate regimes between these two extremes, but the case we consider here is that which has been realised in every blind quantum computing protocol to date. In fact, there are many places where paddings can be applied: the output state that Bob holds after the protocol (the ``output padding''), the input state that Bob holds before the protocol (which is likely already the output of an earlier step), and on the states that Alice and Bob exchange (the ``input padding''). It is only the output padding that we are constraining.

	The padding that we choose comprises $n$-qubit Pauli operators $P_y\in\mathcal{P}_n$ drawn from a probability distribution $\alpha_y$. The $2n$-bit string $y$ may be split into two $n$-bit strings $x$ and $z$ specifying the $X$ and $Z$ components,
	$$
	P_{y} = i^{x^Tz}\bigotimes_{i=0}^n X_i^{x_i}Z_i^{z_i} = i^{x^Tz}X_xZ_z.
	$$
	Previous schemes \cite{childs, broadbentUniversalBlindQuantum2009, fitzsimons2017, giovannetti2013} have used a uniform Pauli padding, where $\alpha_x=\frac{1}{2^{2n}}$. One notable exception, outside the current purview \footnote{This is a case of classical Alice with computational security.}, is \cite{mahadev2020}, where an encrypted controlled-\textsc{not} is achieved with a reduced padding. The main benefit of the Pauli padding is that measurements can be performed by Bob and decoded by Alice. If the padding is $X_xZ_z$ and the measurement outcome is $r$, then the intended measurement result is $r\oplus x$. Furthermore, if Bob is restricted to applying Clifford gates, Alice will be able to track the transformation of the padding through all subsequent steps -- the individual steps, each with Pauli padding output, can be chained together to form a computation.

	Of extensive interest will be the (anti) commutation properties of two operators. We summarise this with
	$$g_ig_j=(-1)^{c(g_i,g_j)}g_jg_i$$
	with values $c(g_i,g_j)=0,1$ conveying commutation and anti-commutation respectively. Given the symplectic matrix $\Omega = \begin{pmatrix}0&\identity_n\\
	\identity_n&0\end{pmatrix}$, the commutation relations between Pauli strings are calculated by
	$$
	c(P_x,P_y) = x^T\Omega y\text{ mod }2.
	$$
	It is straightforward to implement a Clifford gate $C$ on a state with Pauli padding since these are precisely the gates which transform a Pauli operator into another Pauli operator. If we apply the Clifford gate $C$ to a Pauli padded state $P_x\ket{\psi}$, then the outcome we get is
	\begin{align*}
	CP_x\ket{\psi}&=CP_xC^\dagger C\ket{\psi} \\
	&=P_{x'}C\ket{\psi},
	\end{align*}
	our target state $C\ket{\psi}$ with a new padding $P_{x'}=CP_xC^\dagger$. The Clifford property guarantees this to be a Pauli that is easily calculated by anyone who knows $x$ and $C$ (i.e.\ Alice, but not Bob).

	\subsection{Preserved Pauli Subspace}

In \cref{eq:PF}, we introduced the subspace $\mathcal{P}_\mathcal{F}$ that is, up to a phase, unchanged by any of the members of $\mathcal{F}$. Such terms will serve no useful purpose within a padding set, allowing for a reduction in quantum resources for Alice.

		Without loss of generality, we can adjust $\mathcal F$ such that the unitaries always commute with the elements of $\mathcal{P}_\mathcal{F}$. To see this, note that for any $U\in\mathcal{F}$ and $P_x\in \mathcal{P}_\mathcal{F}$, we must have $c(U,P_x)\in\{0,1\}$, and for any Pauli string $P_y$,
		$$
		c(P_yU,P_x)=c(P_y,P_x) \oplus c(U,P_x)
		$$
		so $P_yU$ also shares the same elements of $\mathcal{P}_\mathcal{F}$ \footnote{The implementation of a blind $U$ or $P_yU$ is equivalent because, having implemented one, you can implement the other just by adjusting the padding.}. We can always find a $y$ such that $c(P_yU,P_x)=0$ for all $P_x\in \mathcal{P}_\mathcal{F}$. We therefore take our adjusted $\mathcal{F}$ to be one for which $UP_x=P_xU$ for all $U\in\mathcal{F}$ and $P_x\in \mathcal{P}_\mathcal{F}$. All the members of the updated $\mathcal{F}$ commute with $\mathcal{P}_\mathcal{F}$.

		Consider the Pauli decomposition of $U$, $U=\sum_z\gamma_zP_z$. By definition, for any $P_x\in \mathcal{P}_\mathcal{F}$, $c(U,P_x)=0$, so it must be that $c(P_z,P_x)=0$ for all $x: \gamma_z\neq 0$. This suggests that we should introduce the Pauli subspace $\mathcal{B}$, defined by \cref{eq:B}.	
		If we represent $\mathcal{P}_\mathcal{F}$ by an $r\times 2n$ binary matrix $F$, then $\mathcal{B}$ is represented by the $(2n-r)$-dimensional null space of $F\Omega$. $\mathcal{B}$ has a basis $\langle B_1, \ldots, B_{2n-r}\rangle $ and $B_x = \prod_{i=1}^{2n-r} B_i^{x_i}$ is a general element in $\mathcal{B}$ (this may not be Hermitian). It immediately follows that 
		\begin{lemma}\label{lem:decompose}
		Any $U\in\mathcal{F}$ (for the updated $\mathcal{F}$) has a Pauli decomposition $U=\sum_y \gamma_y B_y$.
		\end{lemma}

There is a phase ambiguity with $\mathcal{P}_\mathcal{F}$, $\mathcal{B}$ as we are not distinguishing $\pm P_x,\pm iP_x$. This will not unduly affect us.
		
		\begin{example}
		To implement a blind Hadamard gate, we consider the set $\mathcal{F}=\{\identity,H\}$, which has $\mathcal{P}_\mathcal{F}=\langle Y\rangle$. Given that $HYH=-Y$, we update $\mathcal{F}\rightarrow\{\identity,XH\}$ such that $XHYHX=Y$.
		The set $\mathcal{P}_\mathcal{F}$ remains unchanged, as does the communication cost.

			The only single-qubit Pauli operator that commutes with $Y$ is $Y$ itself, so $\mathcal{B}=\langle Y\rangle $. We see that $XH=\frac{1}{\sqrt{2}}(\identity+iY)$ is a linear combination of terms from $\mathcal{B}$. 
			\end{example}

			\section{The Blind Quantum Computation Protocol}

			Let $\mathcal{A}_{\mathcal{F}}$ be an interactive protocol between Alice and Bob which we call a \emph{Quantum Gate Protocol}, see \cref{fig:protocol}. Alice aims to get Bob to apply $U^d$ to the state $\rho$ that he initially holds. Both can communicate classically, Bob can perform universal quantum computation and either Alice (PS) or Bob (RM) can prepare and send quantum states. In RM, Alice measures the qubits she receives. In PS, Bob incorporates the qubits into the computation he does. If there is a measurement involved, he sends the results to Alice in the next round of communication. We will explicitly analyse the RM protocol. Only minor alterations are required to apply equivalent reasoning to PS. The conclusions are the same.

			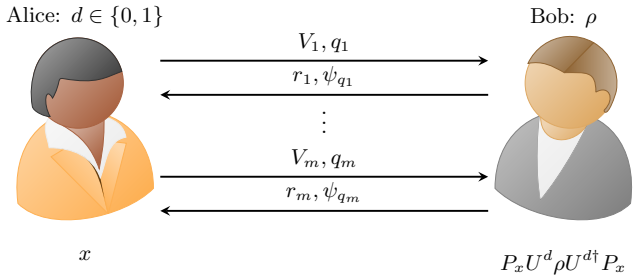
\begin{figure}
			\centering
			\begin{adjustbox}{max width=0.48\textwidth}
			\begin{tikzpicture}[>=stealth]
			\node[alice, minimum size=2cm,label=below:{$x$}, label=above:{Alice: $d \in \{0,1\}$}] (alice) {};
			\node[bob, mirrored, minimum size=2cm, right=5cm of alice, label=below:{$P_x U^d \rho {U^{d\dagger}}P_x$}, label=above:{Bob: $\rho$}] (bob) {};
			\draw[->,thick] (1.1,1) -- (5.9,1) node[midway, above]{$V_1,  q_1$};
			\draw[->,thick]  (5.9,0.5)--(1.1,0.5) node[midway, above]{$r_1, \psi_{q_1}$};

			\node at (3.5,0.2) {\vdots};

			\draw[->,thick] (1.1,-0.7) -- (5.9,-0.7) node[midway, above]{$V_m,  q_m$};
			\draw[->,thick]  (5.9,-1.2)--(1.1,-1.2) node[midway, above]{$r_m, \psi_{q_m}$};
			\end{tikzpicture}
			\end{adjustbox}
			\caption{General Receive \& Measure protocol. Bob starts with his state $\rho$ and additional ancilla states. In each round, Alice sends Bob a message. Bob responds with a classical message and some quantum communication, which Alice measures. By the end, Bob holds an encrypted version of his state with $U^d$ enacted, and Alice knows the encryption key. 
			}
			\label{fig:protocol}
			\end{figure}

			Let  $T_k$ be the transcript of all information Bob has obtained through the first $k$ rounds of the protocol. If no index is given, $T$ refers to all information obtained throughout the protocol. 
			Let $N_i$ be the number of qubits sent in round $i$, so $N = \sum_i N_i$ is the total number of qubits sent during the protocol.

			When the protocol is complete, Bob will have possession of the state $P_x U^d \rho U^{-d} P_x$, where Alice can compute $x$. Bob does not have all the information to compute $x$; based on his transcript $T$, he knows that the padding must be drawn from one of the padding sets $\alpha^{(d)}=\{(\alpha_x,P_x)\}$, where he differentiates based on $d$ since this is what he wants to learn. Our required hiding property is that Bob can gain no advantage in guessing $d \in \{0,1\}$ given his knowledge of the possible final states $\rho'(d) = \sum_x \alpha_x P_x U^d\rho U^{-d} P_x$.

			\begin{lemma}\label{lem:hiding}
			Any blind quantum gate protocol $\mathcal{A}_{\mathcal{F}}$ (with $\identity \in \mathcal{F}$) that has the hiding property must satisfy:
			\begin{align}
			\forall T \forall \rho \quad \sum_{x}\alpha^{(0)}_x P_x \rho P_x = \sum_{x} \alpha^{(1)}_x P_x U \rho U^\dagger P_x \label{hiding property}
			\end{align}
			where $\alpha^{(d)}_x$ are the corresponding distributions for padding sets for the transcript $T$ and gate choice $d$.
			\end{lemma}
			\begin{proof}
			Bob chooses a $\rho$ and follows the protocol, generating some transcript $T$. If the pair $\rho,T$ do not satisfy \cref{hiding property}, the possible states from Bob's perspective are not equal. There exists a measurement that yields a non-zero advantage in distinguishing the states, contradicting the hiding property.
			\end{proof}

			Whilst \cref{lem:hiding} cannot guarantee the security of the scheme (Alice could always send $d$ to Bob classically), it is necessary for Bob to gain no advantage in guessing $d$.

			\subsection{Entropy}

			For Bob to not know what is going on in a system that is entirely under his control, Alice will have to introduce uncertainty, either by sending him qubits in states that he does not have complete knowledge of (PS), or by manipulating the correlations to Bob within states that she is sent (RM). We quantify this uncertainty with the Von Neumann entropy \cite{nielsen00},
			\begin{align*}
			S(\rho) = -\Tr(\rho \log_2(\rho)).
			\end{align*}
			The entropy is invariant under unitary transformations, $S(\rho) = S(U\rho U^\dagger)$. If we were to perform a projective measurement on a state $\rho$, obtaining result $i$ with probability $p_i$, then the system is projected onto the state $\rho_i$. On average, the entropy is non-increasing \cite{lindblad1972}:
			\begin{align}
			S(\rho) \geq \sum_{i} p_i S(\rho_i).\label{eq:entropy_measure}
			\end{align}

			To bound the resources needed by Alice, we follow the average entropy throughout the protocol from Bob's perspective. Entropy is only introduced into the system by Alice measuring in a basis and getting results that are unknown to Bob. This is bounded above by the number of qubits that Bob has sent. After round $k$, transcript $T_k$ arises with probability $p(T_{k})$, leading to Bob describing his state as $\rho_{T_k}$. Bob's expected entropy is
			$$
			s_k=\sum_{T_k} p(T_{k}) S(\rho_{T_k}).
			$$

			\begin{lemma}\label{lem:entropy_RM} 
			For any RM protocol, if Bob sends Alice an expected number of qubits $E(N_k)$ in round $k$, then
			\begin{align}
			s_k+ E(N_{k+1}) \geq s_{k+1}. \label{single_line2}
			\end{align}
			\end{lemma}
			Applied iteratively, this ultimately conveys that
			\begin{align}
			S(\rho) + E(N) \geq \sum_T p(T) S(\rho_T)=s_\text{final},\label{entropy_bound}
			\end{align}
			i.e.\ Bob's uncertainty about his final state derives from any uncertainty in the state he starts from coupled with the uncertainty introduced by giving parts of his correlated states to Alice.

			\begin{proof}
			At the end of round $k$, Bob holds the transcript $T_k$ and a state $\rho_{T_k}$. Let $m$ be the message sent from Alice in round $k+1$. From Bob's perspective, this message occurs with probability $p(m|T_{k})$ and contains the following pieces of information; a unitary $V_m$, a set of measurements to perform, a subset of $n_m$ qubits $A_m$ that Bob must send to Alice. Here we let $B_m$ be the remaining qubits that will remain with Bob.
			\begin{align*}
			s_{k}+ E(N_{k+1})	&= \sum_{T_{k}} p(T_{k})(S(\rho_{T_{k}})+E(N_{k+1}|T_{k}))\\
			&= \sum_{T_{k}} p(T_{k})(S(\rho_{T_{k}})+\sum_m p(m|T_k) n_m)
			\end{align*}
			This message could change Bob's perspective on his current state $\rho_{T_k}$ into a new state $\rho_{T_k, m}$. We must have  $\rho_{T_k} = \sum_{m} p(m|T_k) \rho_{T_k, m}$, unto which we can apply the concavity of the entropy:
			\begin{align*}
			s_k	+ E(N_{k+1})	& \geq \sum_{T_k,m} p(T_k,m)\left( S(\rho_{T_k,m} ) + n_m \right) \\
			&= \sum_{T_k,m}p(T_k,m) (S( V_m \rho_{T_k,m} V_m^\dagger )+n_m).
			\end{align*}
			This includes the unitary $V_m$ that Bob applied at Alice's behest. We must also include the measurements. Let $r$ be the possible measurement result and $\rho_{T_k, m, r}$ be the resulting state. The combination of $T_k,m$ and $r$ becomes the new transcript $T_{k+1}$.
\begin{multline*}
s_k+ E(N_{k+1}) \\
\begin{aligned}
&\geq\sum_{T_k,m}p(T_k,m)\left(\sum_{r} p(r| T_k, m) S( \rho_{T_k, m, r}) + n_m  \right)\\
&= \sum_{T_{k+1}} p(T_{k+1})  \left( S({\rho_{T_{k+1}}})+n_m \right)
\end{aligned}
\end{multline*}

			The final stage of the round is for Bob to send the $A_m$ subspace of the state to Alice. Given that $S(A) \leq n_m$,
			\begin{align*}
			 s_k+ E(N_{k+1})&\geq  \sum_{T_{k+1}} p(T_{k+1})  \left( S(\rho_{T_{k+1}}^{(A_m B_m)})+S(\rho_{T_{k+1}}^{(A_m)}) \right).\\
			 \intertext{The Araki-Lieb inequality $S(AB)+S(A) \geq S(B)$ \cite{araki1970entropy} gives}
			&\geq  \sum_{T_{k+1}} p(T_{k+1})  S(\rho_{T_{k+1}}^{(B_m)})\\
			&= s_{k+1}.
			\end{align*}
			\end{proof}

This proof gives us two key insights into what an optimal protocol must look like. We require that for any step where entropy is non-increasing, the entropy remains constant. For the Araki-Lieb inequality to be tight, states given away by Bob must be maximally entangled to the state he currently holds. For any measurement performed by Bob, information should not be gained. He must measure with a mutually unbiased basis. 
			
			\subsection{Properties of Padding Sets}

The tale of \cref{lem:entropy_RM} is that in sending qubits, Alice is able to create the entropy required in order to obfuscate what she wants Bob to do. This entropy must be mapped into the padding set in a way that satisfies \cref{lem:hiding}. We can now begin to place more restrictions on the properties of such padding sets, with the aim of understanding what is required of any successful protocol so that we can create and recognise one that saturates \cref{lem:entropy_RM}.

			\begin{lemma}\label{lem:hiding_rule}
			If $P_y\in \mathcal{P}_n \setminus \mathcal{P}_\mathcal{F}$, then \footnote{For $|\mathcal{F}|>2$, we might just take $d$ to index which unitary is chosen, with $d=0$ always corresponding to $\identity$.}:
			\begin{align}
			\forall T,d, \quad \sum_x \alpha^{(d)}_x (-1)^{x^T \Omega y} =0. \label{hiding lemma rule}
			\end{align} 
			\end{lemma}

			\begin{example}[ex2]
			In the case of uniform Pauli padding, $\alpha_x=\frac{1}{4^n}$ for all $x\in\{0,1\}^{2n}$. For all $P_y\in\mathcal{P}_n$ (other than $\identity\in \mathcal{P}_\mathcal{F}$), exactly half of the Pauli operators commute, and half anti-commute, so $\sum_x(-1)^{x^T\Omega y}=0$. \cref{lem:hiding_rule} is satisfied.
			\end{example}

			\begin{proof}

			We prove the statement in two parts, firstly for the padding associated with $d=0$, the identity gate, and then for every other $d\neq 0$.

			Let $P_y \in \mathcal{P}_n \setminus \mathcal{P}_{\mathcal{F}}$ and choose $U$ such that $UP_yU^\dagger = \sum_z \lambda_{z} P_z \neq \pm P_y$ with $\lambda_{z} = \frac{1}{2^n}\Tr(P_z U P_y U^\dagger) \in \mathbb{R}$. 
			Since $P_y$ is traceless, and an involution, we must have $\lambda_{0} = 0$ and $\sum_z \lambda_{z}^2 =1$. As $\lambda_{y}\neq \pm 1, \exists z'\notin \{0,y\}$ such that $\lambda_{z'}\neq 0$.

			We now consider the hiding property in the instance where Bob chooses the state $\rho = \frac{1}{2}(\identity + U^\dagger P_{z'} U )$.
			We can use the decomposition $U^\dagger P_{z'} U = \sum_{y'} \lambda'_{y'} P_{y'}$. Using \cref{lem:hiding}, we have:
			\begin{align*}
			\sum_x \alpha^{(0)}_x P_x\left( \sum_{y'} \lambda'_{y'} P_{y'}\right) P_x &= \sum_x \alpha^{(d')}_x P_x P_{z'}P_x.
			\end{align*}
			Commuting the terms though each other resolves to
			\begin{align*}
			\sum_{y'} \lambda'_{y'}P_{y'} \left(\sum_x \alpha^{(0)}_x (-1)^{{y'}^T \Omega x}\right) &= P_{z'}\left(\sum_x \alpha^{(d')}_{x} (-1)^{{z'}^T\Omega x} \right).
			\end{align*}
			Multiplying by $P_y$ and taking the trace yields the conclusion $\sum_x \alpha^{(0)}_x (-1)^{{y}^T \Omega x}=0$ since $\lambda'_y=\lambda_{z'}\neq 0$.

			We can now show that all other padding sets $\alpha_x^{(d)}$, corresponding to any non-$\identity$ $U\in\mathcal{F}$ ($d\neq 0$), must also have the same property.
			This time, Bob will choose $\rho = \frac{1}{2} (\identity + U^\dagger P_y U)$ and let $U^\dagger P_y U = \sum_z \delta_z P_z$. Again,
			\begin{align*}
			\sum_{z} \delta_{z}P_{z} \left(\sum_x \alpha^{(0)}_x (-1)^{{z}^T \Omega x}\right) &= P_{y}\left(\sum_x \alpha^{(d)}_{x} (-1)^{{y}^T\Omega x} \right).
			\end{align*}
			The coefficient of $P_y$ on the left hand side is $0$ so we must have $\sum_x \alpha^{(d)}_{x} (-1)^{{y}^T\Omega x}=0$.
			\end{proof}

			\begin{example}[ex:x]
			Consider the gate set $\mathcal{F}=\{\identity,HS\}$. $HS$ has the action of cycling through the Paulis
			$$
			Z\rightarrow X\rightarrow -Y,
			$$
			so $\mathcal{P}_\mathcal{F}$ is empty aside from the trivial $\identity$. For an initial state $\frac12(\identity+Z)$, we need the paddings $\alpha^{(0)}$ and $\alpha^{(1)}$ to disguise the difference between $\frac12(\identity+Z)$ and $\frac12(\identity+X)$, meaning
			$$
			\sum_x\alpha^{(0)}_x(-1)^{(0,1)\Omega x}=0,\qquad\sum_x\alpha^{(1)}_x(-1)^{(1,0)\Omega x}=0.
			$$
			Repeating for the states $\frac12(\identity+X)$ and $\frac12(\identity+Y)$, we get enough information to solve for all the $\alpha_x^{(d)}=\frac{1}{4}$. The uniform padding is the only Pauli padding option.
			\end{example}

\begin{lemma}\label{lem:blind_padding}
For any blind gate set $\mathcal{F}$, $$\alpha_x =\begin{cases}
			\frac{1}{2^{2n-r}} & P_x \in \mathcal{B}\\ 0 & P_x \notin \mathcal{B}
			\end{cases}  $$ forms a valid padding set, satisfying \cref{lem:hiding_rule}.
\end{lemma}\begin{proof}	
			Consider the case $P_y \notin \mathcal{P}_\mathcal{F}$, which implies that $\exists U = \sum_z \gamma_z B_z \in \mathcal{F}$ such that $UP_y \neq P_yU$. Since each term $B_z$ either commutes or anti-commutes with $P_y$, but they cannot all commute, there is a $z'$ where $B_{z'}$ anti-commutes with $P_y$. This implies exactly half of $\mathcal{B}$ anti-commutes with $P_y$ and we have $\sum_x \alpha_x (-1)^{x^T \Omega y}=0$.
			\end{proof}
			We will see that the use of this padding set in comparison to uniform Pauli padding provides the optimal solution.

			\section{Communication Lower-Bound}

			With \cref{lem:hiding_rule}, we can now lower bound the expected number of qubits for any blind computation protocol.

			\begin{theorem}[Resource bound]\label{thm:resource}
			Let $\mathcal{F}$ be a set of quantum gates acting on $n$ qubits and $r = \dim(\mathcal{P}_\mathcal{F})$. For any set of Pauli-padded blind gate protocols, where Alice has access to either PS or RM, the expected number of qubits of communication is bounded by: 
			\begin{align*}
			E(N) \geq 2n-r.
			\end{align*}
			\end{theorem}

			\begin{proof}
			We prove this by choosing a special $\rho$ for \cref{entropy_bound} for which $S(\rho)=r$, $\rho_T = \frac{1}{2^{2n}} \identity_{2n}\ \forall T$ and we choose $d=0$. Specifically, let $\{P_i'\}_{i=1}^r$ be a basis of $\mathcal{P}_\mathcal{F}$, and extend this basis to the full $n$ qubit space, $\{P_{i}'\}_{i=1}^{2n}$. We define the $2n$-qubit state
			$$
			\rho=\frac{1}{2^{2n}}\prod_{i=r+1}^{2n}(\identity_{2n}+P_i'\otimes P_i'),
			$$
			which has $S(\rho)=r$ by design.

			Alice and Bob follow the protocol, with Bob applying the actions specified by Alice on the first $n$ qubits.
			At the end of the protocol with transcript $T$, the state from Bob's perspective is
			\begin{align*}
			\rho'_T
			&= \frac{1}{2^{2n}} \sum_{y} P_y' \otimes P_y' \left( \sum_x \alpha_x (-1)^{y^T \Omega x} \right) .
			\end{align*}
			Any non-$\identity$ term is of the form $P_y'\otimes P_y'$ with $P_y'\notin \mathcal{P}_\mathcal{F}$, so we can apply \cref{lem:hiding_rule}, leaving only
			\begin{align*}
			\rho'_T&= \frac{1}{2^{2n}} \identity_{2n}  \left( \sum_x \alpha_x (-1)^{\vec{0} \Omega x} \right) \\
			&= \frac{1}{2^{2n}} \identity_{2n}.
			\end{align*}
			Consequently, $\forall T, \, S(\rho_T) = 2n$ and using \cref{entropy_bound}, we have
			$
			r+E(N) \geq 2n
			$, as desired.
			\end{proof}

			Typically in blind quantum computation protocols \cite{broadbentUniversalBlindQuantum2009, giovannetti2013}, Alice is restricted to sending separable states. Our proof does not impose this restriction. This optimality may require Alice to prepare arbitrary (entangled) quantum states for PS protocols or measure in arbitrary bases for RM protocols, on up to $2n-r$ qubits.

			\section{Communication-Optimal Blind Gate Protocols}
			We have shown that for any given gate set $\mathcal{F}$ on $n$ qubits, any blind quantum protocol that hides the action of any gate in $\mathcal{F}$ under a Pauli padding, requires that Bob sends at least $2n-r$ qubits in the process. 
			In this section, we will show that this bound can be saturated.

			\begin{definition}\label{def:standard}
			The \emph{standard transformation} of $\mathcal{B}$ is a Clifford unitary $V_\textnormal{st}$ which acts as \footnote{We have assumed here that $n\leq 2n-r$. Otherwise, the additional $\otimes\identity$ term moves from the LHS of the equation to the RHS. This does not affect the results.}
			$$
			V_\textnormal{st}\left(B_i\otimes\identity \right)V_\textnormal{st}^\dagger=X_iZ_{c_i}
			$$
			where $c_i=[c_{i,1},c_{i,2},\ldots c_{i,i-1},0,0,\ldots 0]$ has $2n-r$ elements with $c(B_i,B_j)=c_{i,j}$.
			\end{definition}
		The $Z_{c_i}$ terms are chosen in such a way that the (anti) commutation properties of the $B_i$ are preserved and, moreover, that $V_\textnormal{st}B_xV_\textnormal{st}^\dagger\ket{0}^{\otimes(2n-r)}=\ket{x}$ since all the $Z$ terms of a given qubit $i$ appear to the right of the $X_i$ and have a trivial action on $\ket{0}$.
		
		\begin{example}
		The set $\mathcal{F}=\{\identity,HS\}$ has a trivial $\mathcal{P}_\mathcal{F} = \{\identity\}$. Everything commutes with this, so we can select, for instance, $\mathcal{B}=\langle X,Z\rangle$. In this case, the standard transformation has
		$$
		V_\textnormal{st}(B_1\otimes\identity)V_\textnormal{st}^\dagger=X_1,\qquad V_\textnormal{st}(B_2\otimes\identity)V_\textnormal{st}^\dagger=Z_1X_2
		$$
		and can be implemented by the Clifford circuit
		
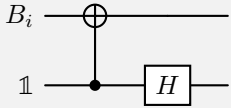
\[
		\begin{quantikz}
		\lstick{$B_i$}& \targ{} & & \\
		\lstick{$\identity$}&\ctrl{-1} & \gate{H} &
		\end{quantikz}
		\]
		\end{example}

		\begin{lemma}\label{lemma:basis}
		Let $U = \sum_z \gamma_z B_z$ be a unitary and $\ket{\phi_{U}} = \sum_z \gamma_z \ket{z} = V_{\textnormal{st}} U V_{\textnormal{st}}^\dagger \ket{0}^{\otimes 2n-r}$ be the corresponding state. $\{\ket{\phi_{B_x U} }\}$ forms an orthonormal basis.
		\end{lemma}
		
		\begin{proof}
\begin{align*}
		\delta_z
		&= \Tr(B_zUU^\dagger)/2^n\\
		&= \Tr\left(\sum_{a,b} \gamma^*_a\gamma_b B_zB_bB_a^\dagger\right)/2^n\\
		&=\sum_{b}  \gamma^*_{b\oplus z} \gamma_b \Tr(B_zB_bB_{z\oplus b}^\dagger)/2^n\\
		\intertext{Since $B_{z\oplus b}=B_zB_b(-1)^{b^T c_z}$ where $c_z =\bigoplus_iz_ic_i$, we have}
		&= \sum_{b}  \gamma^*_{b\oplus z} \gamma_b \Tr(B_z B_b B_b^\dagger B_z^\dagger)(-1)^{b^T c_z}/2^n \\
		&= \sum_{b} \gamma^*_{b\oplus z} \gamma_b (-1)^{b^T c_z  }.
		\end{align*}
Now observe that
$$\ket{\phi_{B_xU}}=V_{\textnormal{st}} B_x V_{\textnormal{st}}^\dagger\ket{\phi_{U}}=X_xZ_{c_z}\ket{\phi_{U}}.$$
Hence
		\begin{align*}
		|\braket{\phi_{B_x U} }{ \phi_{B_y U} }|
		&= |\bra{ \phi_U } Z_{c_x}X_xX_{y} Z_{c_y}  \ket{\phi_U}|\\
		&=|\bra{ \phi_U } X_zZ_{c_z}  \ket{\phi_U}|,
		\intertext{where $z = x \oplus y$. Expanding this explicitly, we have}
		&= \sum_{a} \gamma_a^* \bra{a} X_z Z_{c_z} \sum_b \gamma_b \ket{b}\\
		&= \sum_{b} \gamma^*_{b\oplus z} \gamma_b (-1)^{b^T c_z  }\\
		&=\delta_{x,y}.
		\end{align*}
		\end{proof}
		
		An immediate corollary of \cref{lemma:basis}, is that $\left\{\ket{\phi^*_{B_xU}} = \sum_x \gamma^*_x \ket{x}\right\}$ also forms an orthonormal basis.
		
		This now gives us the tools to prove that \cref{general_RM_protocol} is a communication-optimal RM protocol for any gate set $\mathcal{F}$. A similar PS communication optimal protocol is given in \cref{app:PS}.

		\begin{theorem}\label{thm:opt_prot}
		\cref{general_RM_protocol} is an optimal RM protocol for $\mathcal{F}$, where Bob sends $2n-r$ qubits to Alice.
		\end{theorem}
		
		\begin{proof}
		Let $\ket{\phi_{B_zU}^*} = \sum_x \gamma_x^* \ket{x}$ where $\sum_x \gamma_x B_x = B_z U$. Bob's circuit in \cref{fig:protocol} produces an output state $\ket{\Psi'}=(V\otimes\identity)\ket{\Psi}=\frac{1}{\sqrt{2^{2n-r}}} \sum_x V \ket{x}_A B_x\ket{\psi}_B $.

		Alice is given the first $2n-r$ qubits of $\ket{\Psi'}$, and performs her measurement. Suppose the measurement result $z$ is obtained, and Alice's state is projected onto $V\ket{\phi^*_{B_zU}}$. The resulting state (up to normalisation) is:
		\begin{align*}
		\left(\bra{\phi^*_{B_zU}} V^\dagger\right) \otimes \identity_B \ket{\Psi'}&=
		\left(\bra{\phi^*_{B_zU}} V^\dagger V\right) \otimes \identity \ket{\Psi}\\
		&=\left(\sum_y \gamma_y \bra{y}\otimes \identity \right) \frac{\sum_x \ket{x} B_x \ket{\psi}}{\sqrt{2^{2n-r}}}\\
		&= \frac{1}{\sqrt{2^{2n-r}}}\sum_x \gamma_x B_x \ket{\psi}\\
		&=  \frac{1}{\sqrt{2^{2n-r}}} B_z U\ket{\psi}.
		\end{align*}
	
	Each measurement result $z$ is equally likely, resulting in the unitary $B_zU$. Bob does not know Alice's chosen measurement basis, nor her measurement result. From his perspective, the state he holds is $\rho_B = \Tr_A(\proj{\Psi'})$, which is independent of $U$ and leaks no information, meaning Bob can gain no advantage in guessing $U$.
		\end{proof}

		\begin{example}[ex3.1]
		Consider the target gate set $\mathcal{F}=\{R_z(\theta) | \theta \in [0, 2\pi )\}$, where
		$$
		R_z(\theta)=\begin{bmatrix} 1 & 0 \\ 0 & e^{i\theta} \end{bmatrix}=e^{i\theta/2}\left(\cos\frac{\theta}{2}\identity-i\sin\frac{\theta}{2}Z\right).
		$$
		With $\mathcal{P}_{\mathcal{F}}=\langle Z\rangle=\mathcal{B}$, \cref{thm:resource} conveys that we need at least 1 qubit of communication. We can choose $V=\identity$. Bob applies the circuit ($V=\identity$)
		$$
		\begin{quantikz}
		\lstick{\ket{+}} & \ctrl{1}&\\
		\lstick{\ket{\psi}} & \gate{Z} &
		\end{quantikz}
		$$
		sending the first qubit to Alice.
		The resulting state is
		$$
		\ket{\Psi } = \frac{1}{\sqrt{2}}\left( \ket{0}_A \otimes \identity \ket{\psi}_B +\ket{1}_A \otimes Z \ket{\psi}_B \right).
		$$
		Having chosen a target $\theta'$, Alice measures in the basis
		\begin{align*}
		\ket{ \phi_{0} } &= \cos \frac{\theta'}{2}\ket{0} + i\sin\frac{ \theta'}{2} \ket{1} \\
		\ket{\phi_1} &= X \ket{\phi_0}
		\end{align*}
		where $\braket{\phi_i}{\phi_j}=\delta_{ij}$. We see that
		\begin{align*}
		(\bra{\phi_i} \otimes \identity )\ket{\Psi} &=  e^{i \theta'/2} \frac{1}{\sqrt{2}} Z^iR_z(\theta') \ket{\psi}.
		\end{align*}
		Bob holds the state $Z^i R_z(\theta')\ket{\psi}$ and Alice knows the measurement outcome, bit $i$, with both values having been equally likely.

		Bob, who has no idea of bit $i$, sees the state as one of
		$$\frac12(\proj{\psi}+Z\proj{\psi}Z),
$$
		independent of $\theta$. There is nothing he can do to discover $\theta'$.
		\end{example}

		\section{Separable States for Clifford Gates}

		Our optimal protocol for $\mathcal{F}$ generically requires Alice to measure in an entangled basis. In the case where $\mathcal{F} = \{\identity, U_\textnormal{Cl}\}$ for any Clifford gate $U_\textnormal{Cl}$, we will now show how to choose $V$ to ensure that Alice only needs to measure separable stabilizer states.
		
		If $U$ is a Clifford gate, then the state $\ket{\phi_{B_zU}^*}$ is created by Clifford gates, and must be a stabilizer state. Moreover, the stabilizers are independent of $z$ (only the signs on the stabilizers depend on $z$). This observation already reduces Alice's role to measuring stabilizers rather than an arbitrary basis. Let $\langle g_i\rangle$ and $\langle h_i\rangle$ be bases for the stabilizers of $\ket{\phi_{\identity}}$ and $\ket{\phi_{U}}$. Alice must then measure the stabilizers $Vg_iV^\dagger$ or $V h_i V^\dagger$. We also choose to fix $V$ as Clifford. The goal of this section is to prove that we can choose a $V$ such that the generators are $\langle Z_i\rangle$ and $\langle X_i\rangle$ respectively for the two gates.

		\begin{lemma}\label{lem:stabsRM}
		Let $U_\textnormal{Cl}$ be a Clifford gate and $\mathcal{F}=\{\identity,U_\textnormal{Cl}\}$. For any optimal blind gate protocol $\mathcal{A}_{\mathcal{F}}$, described by \cref{general_RM_protocol}, the states $\ket{\phi_\identity}$ and $\ket{\phi_{U_\textnormal{Cl}}}$ are stabilizer states that do not share any stabilizers.
		\end{lemma}
		\begin{proof} Assume that the two states share a stabilizer $S \in \langle g_1, \ldots ,g_{2n-r} \rangle \cap \langle h_1 \ldots, h_{2n-r} \rangle\setminus\identity$. We adapt \cref{general_RM_protocol} by choosing $V$ such that $VSV^\dagger=Z_1$. We can always update the presentation of the stabilizers so that all other generators are $\identity$ on the first qubit.

		In \cref{general_RM_protocol}, Alice always measures the stabilizer $Z_1$ on the qubit she receives. We take advantage of this, with Bob instead measuring the first qubit in the $Z$ basis and sending Alice the remaining qubits along with the measurement result. Alice is then responsible for measuring the remaining stabilizers.

		Clearly, this still constitutes a blind protocol quantum gate protocol for $\mathcal{F}=\{\identity, U_{\text{Cl}}\}$ which only requires $2n-r-1$ qubits of communication, contradicting \cref{thm:resource}. The assumption must be false.
		\end{proof}
		
		\begin{lemma}\label{lem:sep_stab}
		For any Clifford gate $U_\textnormal{Cl}$, there exists an optimal blind gate protocol for $\mathcal{F} = \{\identity, U_\textnormal{Cl}\}$ such that Alice only needs to measure $X$ and $Z$ stabilizers. 
		\end{lemma}
		
		\begin{proof}

		Let $\{g_{i}\}$ and $\{h_i\}$ be the stabilizers Alice must measure for $U = \identity$ and $U=U_{\text{Cl}}$, respectively. By \cref{lem:stabsRM}, the matrix $M_{i,j} = c(g_j, h_i)$ must be invertible. Row operations can deliver a new basis $\{\hat{h}_i\}$ such that $c(g_i, \hat{h}_j) = \delta_{i,j}$. There must exist a Clifford $V$ such that $V^\dagger g_i V = Z_i$ and  $V^\dagger \hat{h}_i V = X_i$. For this $V$, Alice only needs to measure separable stabilizer states.
		\end{proof}
		
		An explicit formula for $V$ can be computed. We know $\ket{\phi_{\identity}} = V_\textnormal{st} \identity V_\textnormal{st}^\dagger \ket{0\ldots 0}$, so has stabilizers $Z_1, Z_2, \ldots Z_{2n-r}$. Let $\hat{h}_i$ be a basis for the stabilizers for $\ket{\phi_{U}}$, satisfying $c(Z_i, \hat{h}_j) = \delta_{i,j}$. This can only be achieved if $\hat{h}_i = X_i Z_{f_i}$, for some binary vector $f_i$. We want our $V$ to satisfy $VZ_iV^\dagger = Z_i$ and $VX_i Z_{f_i} V^\dagger = X_i$. Combining these together yields $VX_iV^\dagger  = X_i Z_{f_i}$ which is achieved using $S_i$ if $f_{i,i}=1$ and $\text{c}Z_{i,j}$ if $f_{i,j}=1$.

		\begin{example}[ex:CP]
		The Pauli decomposition of controlled-phase, c$Z$, is
		$$\text{c}Z = \frac{1}{2}\left( \identity + Z_1  + Z_2 - Z_1Z_2 \right).$$
		To achieve a blind implementation of the family $\mathcal{F} = \{\identity$, c$Z\}$, we take $\mathcal{P}_{\mathcal{F}} = \langle Z_1, Z_2 \rangle=\mathcal{B}$.		

The states (corresponding to measurement bases) for the two different gates are
\begin{align*}
\ket{\phi_\identity} &= \ket{00} \\
\ket{\phi_{\text{c}Z}}&= \frac{1}{2}(\ket{00}+\ket{01}+\ket{10}-\ket{11})
\end{align*}
with stabilizers $\langle Z_1,Z_2\rangle$ and $\langle X_1Z_2,Z_1X_2\rangle$ respectively. We can simultaneously transform both bases into separable bases simply by applying a $V$ that is controlled-phase.
		Hence the circuit for a scheme with separable qubits is:
		
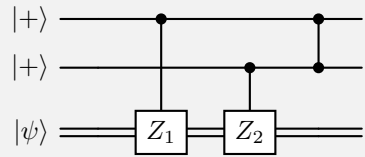
\[
		\begin{quantikz}[wire types={q,q,c}, classical gap = 0.05cm]
		\lstick{$\ket{+}$}&& \ctrl{2} && \ctrl{1}&\\
		\lstick{$\ket{+}$}&&& \ctrl{1}&\ctrl{0} &\\
		\lstick{$\ket{\psi}$}&& \gate{Z_1}&\gate{Z_2}&&
		\end{quantikz}
		\]

		If Alice measures both qubits in the $Z$ basis, and obtains the bit-string $x$, the resulting state is $Z_x \ket{\psi}$. If instead Alice measures in the $X$ basis, the resulting state would be $Z_x \text{c}Z\ket{\psi}$.

		The PS protocol is almost identical,
		
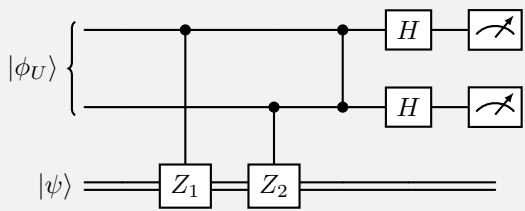
\[
		\begin{quantikz}[wire types={q,q,c}, classical gap = 0.05cm]
		\lstick[2]{\ket{\phi_U}}&& \ctrl{2} && \ctrl{1}& \gate{H} & \meter{}\\
		&&& \ctrl{1}&\ctrl{0} &\gate{H} & \meter{}\\
		\lstick{$\ket{\psi}$}&& \gate{Z_1}&\gate{Z_2}&&&
		\end{quantikz}
		\]
		with the final padding being a function of the padding on the separable input states $\ket{\phi_U}$ ($Y_x\ket{00}$ or $Y_x\ket{++}$) that Alice sends to Bob, and Bob's measurement result.
		\end{example}

		An alternative procedure that achieves a similar results is to view the Clifford gates as a product of transvections, $U_P= \frac{1}{\sqrt{2}}\left(\identity + i P\right)$ \cite{pllaha2021,o1974lectures}. Each transevection can be implemented with RM (PS) using a single qubit of communication. This gives a straightforward separable protocol in terms of either $2n-r$ or $2n-r+1$ qubits. However, with some effort, one can transform the one extra qubit (if present) into an equivalent of $V$.

		\section{Conclusions \& Open Problems}
		
		In this paper, we have used entropy techniques to lower bound the number of qubits of communication required between client and server in any information-theoretically secure blind quantum computation scheme (that encrypts states under Pauli padding), and have given corresponding protocols that saturate these bounds. Any set of gates, $\mathcal{F}$, on $n$ qubits can be realised with exactly $2n-r$ qubits of communication. In this optimal protocol, Alice does most of the work. Bob essentially provides a quantum memory with Clifford operations. If we consider Alice's measurement as a unitary followed by computational basis measurements, then not only can that unitary be entangling, if the target unitary is in the $k\textsuperscript{th}$ level of the Clifford hierarchy, then so is the unitary that Alice implements. Alice has most of the burden of the complexity, despite being the semi-classical participant!

		We have described two different protocols -- prepare and send, and receive and measure. PS is the more commonly explored protocol. However, our RM variant is more natural in many ways. The security proof is much clearer as Bob never receives any communication from Alice. She does the measurements and keeps both the basis and the results secret. There can be no leakage. The method may also have better noise tolerance properties. We imagine that Bob is presenting a perfect computational device to Alice by operating an error correcting code on top of a much larger and noisier physical computer. With PS, Alice has all the burden of preparing her states in a large, complex error correcting code in order to survive the journey to Bob. For RM, the reverse direction is much simpler as the error correction is already built in. Alice just has to process that by incorporating the decoding process into her measurements.
		
		We have concentrated specifically on the question ``given a fixed set of target gates, what's the smallest amount of communication possible?''. This naturally leads to demanding how many other gates can also be implemented, having found the minimal communication. For Clifford gates, using separable stabilizer measurements, we will show in a future work that this number is bounded between $2^{2n-r}$ and $3^{2n-r}$.
		
		\begin{example}[additonal_gates]
		Following \cref{ex:CP} for $\mathcal{F}=\{\identity,cZ\}$, what other gates can we add to the set for no additional cost?
		Each of the 6 measurement bases below realises a secure implementation of a distinct Clifford gate.

		\begin{center}
		\begin{tabular}{ccc}
		\toprule
		Stabilizers & $U$ \\
		\midrule
		$\langle Z_1,Z_2\rangle $ & $\identity$ \\
		$\langle Z_1,Y_2\rangle $ & $S_2$ \\
		$\langle Y_1,Z_2\rangle $& $S_1$ \\
		$\langle X_1,Y_2\rangle $ & $\text{c}Z S_2$ \\
		$\langle Y_1,X_2\rangle $ & $\text{c}Z S_1$ \\
		$\langle X_1,X_2\rangle $ & $\text{c}Z$ \\
		\bottomrule

		\end{tabular}
		\end{center}

\end{example}

The contrasting extreme of optimality is to fix the resources, say $n$ qubits of communication, and demand how many different unitaries can be implemented (with no regard for what those unitaries are). This was the focus of \cite{mantri2013}, where they focused on PS protocols and other variants. Let $\mathcal{A}_{\mathcal{M}}$ be any protocol in which Alice sends $n$ qubits of communication and can realise all quantum gates that lie on a manifold $\mathcal{M}$. These act on at most $L$ qubits. They were interested in the quantity $\Gamma(n) = \max_{\mathcal{A}_\mathcal{M}} \dim(\mathcal{M})$, ultimately showing that
$$
2^{n}-1\leq\Gamma(n)\leq 2(2^n-1).
$$
Here, $\mathcal{F} = \{ U= U(\theta) \in \mathcal{M} \}$. It must then be the case that $n \geq  2L - \dim(\mathcal{P}_\mathcal{F})$. All gates in $\mathcal{M}$ must have support on $\mathcal{B}$ where $\dim(\mathcal{B}) = 2L - \dim(\mathcal{P}_\mathcal{F}) \leq n$. The manifold of all unitaries (up to a phase) with support on $\mathcal{B}$ must be a submanifold of $\mathcal{M}' = \{U(\theta) = e^{i \sum_{x\neq \vec{0}} \theta_x B_x }, \theta \in \mathbb{R} ^{2^n-1}\}$. We must have $\Gamma(n) \leq 2^{n}-1$. This paper and \cite{mantri2013} demonstrated different protocols that saturate this bound; we in fact have $\Gamma(n) = 2^{n}-1$.

In order to achieve a blind implementation of an $N$ dimensional manifold using separable states/measurements, we can decompose the target gates into a sequence of 1-dimensional manifolds. We can use this notion to improve on the efficiency rating of the original brick work state. In the standard brick work state, every $4$ qubits Alice sends corresponds to a quantum gate specified by $3$ real parameters. This tells us that the standard MBQC protocol is within a factor of $\frac{4}{3}$ of optimality. More recently, \cite{ma2024} suggested an improved version of the brick work state gets within a factor of $\frac{5}{4}$ of optimality.

Our proofs have concentrated solely on encryption provided by Pauli padding. To our knowledge, all blind quantum computing schemes use some flavour of Pauli padding sets. There is no \emph{a priori} reason that blindness necessitates Pauli padding. Adjusting the definition of what gates are included in the padding set could be an option to circumvent our lower bounds, and will be interesting to investigate in the future.

It is also of interest how these results translate into the setting where a classical Alice uses post-quantum cryptography to generate blind gates. Mahadev \cite{mahadev2020} showed that an encrypted controlled-\textsc{not} can be achieved with a single query to a Learning-With-Errors (LWE) circuit, suggesting there are at least two bits of useful computational entropy within a single LWE sample. Could there be more computational entropy that can be leveraged to build more complex families of blind gates for the same number of queries?

		\bibliography{main_PS_RM}
		
 		\appendix\onecolumngrid
 		\crefalias{section}{appendix}
\section{Prepare and Send approach}\label{app:PS}

Here we discuss how our results for RM translate into the PS setting.
The previous entropy bound \cref{lem:entropy_RM} also holds when Alice is the party sending quantum states. The proof is almost identical, except that we consider $n_m$ to be an upper bound for the entropy of the state provided by Alice. All other actions in the protocol can only decrease entropy from Bob's perspective.

The qubit communication bound can also be saturated. For a general family $\mathcal{F}$, corresponding to a basis $\mathcal{B} = \langle B_1 \ldots B_{2n-r} \rangle$, we further need to define an additional Pauli space $\mathcal{Q}$. This has a basis $ \{ Q _1 \ldots Q_{2n-r}\}$ where $c(Q_i,B_j)=\delta_{ij}$, i.e.\ $Q_i$ anti-commutes with $B_i$ but commutes with all other $B_j$, such that $c(Q_y,B_x)=x^Ty$.	

\begin{figure}
	\centering
	\begin{quantikz}[wire types = {q,q,n,q,b}, classical gap = 0.07cm]
		\lstick[4]{$\ket{\phi_{U}}$} & \ctrl{4} &&\gate[style={draw=none}]{\ldots} && \gate{H}& \meter{}\rstick[4]{\hspace{0.7cm} $y$}\\
		&& \ctrl{3} & \gate[style={draw=none}]{\ldots}&& \gate{H}& \meter{}\\[-0.2cm]
		& &  & \gate[style={draw=none}]{\ddots}  & &&\vdots    \\[-0.2cm]
		&& & \gate[style={draw=none}]{\dots}  &\ctrl{1}& \gate{H}& \meter{} \\
		\lstick{$\ket{\psi}$} & \gate{B_1} & \gate{B_2} & \gate[style={draw=none}]{\ldots} & \gate{B_{2n-r}}\slice{1} &\slice{2} &&
	\end{quantikz}
	\caption{General circuit used for optimal blind quantum computing where Alice is capable of PS. Alice prepares the state $\ket{\phi_U}$, which is sent to Bob. The state is input into the circuit with Bob's state $\ket{\psi}$. }
	\label{general_circuit_PS}
\end{figure}
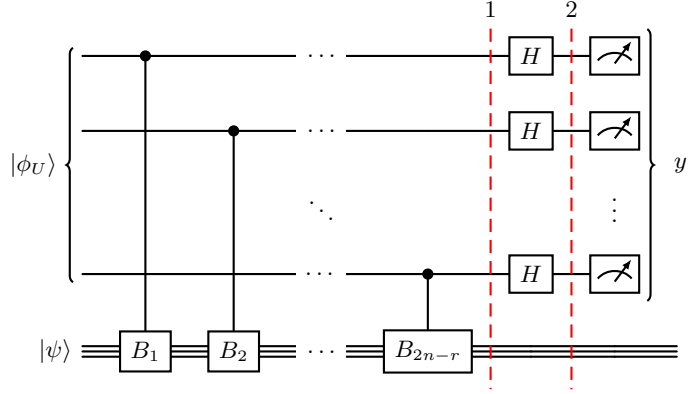

We can now offer a general two round protocol for Alice with PS.

	 	\begin{protocol}\label{prot:main}
	 		Alice plans to blindly implement the gate $\hat{U}\in\mathcal{F}$. She extends $\mathcal{F}$ to $\mathcal{F}'=U(\mathcal{B})$, the group of all unitaries which has support on $\mathcal{B}$ \footnote{It is not necessary to extend $\mathcal{F}$ to $U(\mathcal{B})$; it is sufficient to extend to it a large enough set such that $\mathcal{F}$ is a group that is closed under $U \rightarrow B_x  Q_y U Q_{y}^\dagger , \forall x, y$.}.
	 		\begin{enumerate}
		 			\item Alice selects a $U\in\mathcal{F}'$ uniformly at random. This has a decomposition of $U = \sum_{z} \gamma_z B_z$ using \cref{lem:decompose}.
		 			\item Alice creates the $2n-r$ qubit state$$\ket{\phi_U} =\sum_z \gamma_z \ket{x}=V_\textnormal{st}UV_\textnormal{st}^\dagger\ket{0}^{\otimes(2n-r)}$$ which she sends to Bob.
		 			\item Bob runs the circuit shown in \cref{general_circuit_PS}, getting measurement result $y$, which he sends to Alice.
		 			\item Alice uniformly samples $x \in \{0,1\}^{2n-r}$ and (classically) computes $\Lambda =  B_x \hat{U}Q_y^\dagger  U^\dagger Q_y$, which she sends to Bob.
		 			\item Bob applies $\Lambda$.
		 		\end{enumerate}
	 	\end{protocol}
	 	
	 		 	\begin{theorem}\label{thm:protocol}
	 		 		The 2-round protocol specified in \cref{prot:main} causes Bob to blindly implement a gate $\hat U\in\mathcal{F}$, with Alice sending $2n-r$ qubits to Bob.
	 		 	\end{theorem}
	 	 	\begin{proof}
	 		 		We first prove correctness of the protocol and then show the protocol gives away no advantage to Bob in guessing $\hat{U} \in \mathcal{F}$.	 		
	 		 		At step 3 of \cref{prot:main}, the first slice of \cref{general_circuit_PS}, the state would be $\sum_z \ket{z} \gamma_z B_z\ket{\psi}$. At the second slice, we have:
	 		 		\begin{align*}
	 			 			\ket{\Psi}&= \sum_y \ket{y} \left(\sum_z (-1)^{y^T z} \gamma_z B_z \right)\ket{\psi}
	 			 		\end{align*}
	 		 		Using the key properties of the $\mathcal{Q}$ basis, this can be expressed as
	 		 		\begin{align}
	 			 			\ket{\Psi}&= \sum_y \ket{y} \left(\sum_z  \gamma_z  Q_y B_z Q^{\dagger}_y \right)\ket{\psi}\nonumber\\
	 			 			&= \sum_y \ket{y} Q_y U Q_{y}^\dagger \ket{\psi}.\label{eq:Qpadding}
	 			 		\end{align}	
	 		 		Upon receiving measurement result $y$ in step 4, Alice knows that Bob has state $Q_{y}UQ_{y}^\dagger \ket{\psi}$. This random unitary $Q_{y}UQ_{y}^\dagger\in\mathcal{F}'$ effectively serves as a one-time pad on unitaries. When Bob applies $\Lambda$, he holds the state $\Lambda \ket{\psi} = B_x \hat{U} \ket{\psi}$.

	 	 	Proving security for this protocol is more involved: Bob can now behave adaptively and can effectively query Alice about the state he has received. We consider a general Bob who can deviate in the protocol in any way. He receives two pieces of information: the state $\ket{\phi_{U}}$ and the classical description of the gate $\Lambda$. We will show that the information in $\Lambda$ is entirely independent of the information in $U$. To that end, assume that Bob knows $\hat U$. He still doesn't entirely know $U$ because of the unknown $x$ introduced by Alice into $\Lambda$. He can thus describe $U$ as
	 	 	\begin{equation}\label{eq:U_of_x}
	 		 		U(x)=Q_y^\dagger \Lambda^\dagger B_x\hat UQ_y.
	 		 	\end{equation}
	 	 	Here, $y$ is the data that Bob has returned to Alice. Since Bob is not necessarily following the protocol, he can return anything he wants, which need not necessarily be related to any measurement on the state $\ket{\phi_U}$.
	 	
	 	 	The state that Bob receives is
	 	 	$$
	 	 	\rho_{\hat{U}}=\frac{1}{2^{2n-r}}\sum_x\proj{\phi_{U(x)}}.
	 	 	$$
	 	 	We claim that the $\ket{\phi_{U(x)}}$ defines an orthonormal basis, so $\rho_{\hat{U}}$ is maximally mixed state. Any experiment where Bob uses as input $\rho_{\hat{U}}$ and guesses $\hat{U}$ must produce results that are independent of $\hat{U}$ and yield no advantage. To prove this, we must evaluate $\braket{\phi_{U(z)}}{\phi_{U(x)}}$. It must be true that for any state $\ket{\psi}$,
	 	 	\begin{align*}
	 		 		\braket{\phi_{U(z)}}{\phi_{U(x)}}&=\braket{\phi_{U(z)}}{\phi_{U(x)}}\braket{\psi}{\psi} \\
	 		 		&=\frac{1}{2^{2n-r}}\sum_t\bra{\psi}Q_t^\dagger U^\dagger(z)U(x)Q_t\ket{\psi}
	 		 	\end{align*}
	 	 	where we have applied the unitary circuit of \cref{general_circuit_PS} up to the final slice (unitaries preserve inner products). If we substitute \cref{eq:U_of_x}, this reduces to
	 	 	$$
	 	 	\braket{\phi_{U(z)}}{\phi_{U(x)}}=\frac{1}{2^{2n-r}}\sum_t\bra{\psi}Q_tQ_y^\dagger\hat U^\dagger B_z^\dagger B_x \hat UQ_yQ_t^\dagger\ket{\psi}.
	 	 	$$
	 	 	In the case $x=z$, this is clearly 1; the state is normalised. In the case where $x\neq z$, we write $\hat U^\dagger B_z^\dagger B_x \hat U=\sum_s\eta_sB_s$, which means
	 	 	\begin{align*}
	 		 		\braket{\phi_{U(z)}}{\phi_{U(x)}}&=\frac{1}{2^{2n-r}}\sum_{s,t}(-1)^{(y\oplus t)\cdot s}\eta_s\bra{\psi}B_s\ket{\psi} \\
	 		 		&=\eta_0\bra{\psi}B_0\ket{\psi}.
	 		 	\end{align*}
	 	 	However,
	 	 	\begin{align*}
	 		 		2^{2n-r}\eta_0&=\Tr(\hat U^\dagger B_z^\dagger B_x \hat U) \\
	 		 		&=\Tr(B_z^\dagger B_x) \\
	 		 		&=0.
	 		 	\end{align*}
	 	 	We conclude that
	 	 	$\braket{\phi_{U(z)}}{\phi_{U(x)}}=\delta_{x,z}$ and therefore have $\rho_{\hat{U}} = \frac{1}{2^{2n-r}}\identity_{2n-r}$.
	 	 		 		 	\end{proof}

	 	The need for the second round of communication could be removed if $\forall y, \exists z \, s.t.\, Q_y U Q_y^\dagger U^\dagger\in\mathcal{P}_n$, which is immediately satisfied if $U$ is Clifford.
	 	To implement a gate $\hat{U}$ in this setting, Alice chooses a random $x$ and creates $\ket{\phi_{B_x\hat{U}}}$ which she sends to Bob. We know (\cref{lemma:basis}) these form an orthonormal basis, and leak no information to Bob. Upon receiving measurement result $y$, the resulting state is $Q_y B_x \hat U  Q_y^\dagger \ket{\psi} = P_z \hat U\ket{\psi}$, for some $z$ that can be computed by Alice. Similarly to the the RM setting, we can reduce Alice's need for entanglement generation. Alice could instead send the state $V \ket{\phi_{B_x \hat U}}$ and have Bob immediately apply $V^\dagger$ before continuing as before. We have seen in \cref{lem:sep_stab} that $V$ can be chosen such that $V\ket{\phi_{B_x \identity}}$ is stabilized by $\langle \pm Z_i \rangle $. If $\hat U$ is Clifford, $V\ket{\phi_{B_x \hat U}}$ is stabilized by $\langle  \pm X_{i} \rangle$. To implement the gate $\hat U^d$, Alice would send need to send Bob $Y_x {H^{\otimes (2n-r)} }^d \ket{0 \ldots 0}$, and only needs to send single qubit stabilizer states. See \cref{ex:CP} for a comparison between the single-round PS and the RM methods.

\end{document}